\documentclass[12pt ,onecolumn]{IEEEtran}

\usepackage{times, epsfig}
\usepackage{amssymb}
\usepackage{amsmath}
\usepackage{graphicx}
\usepackage[latin1]{inputenc}
\usepackage{graphicx}
\usepackage{setspace}
\usepackage{psfrag}
\usepackage{cite}
\usepackage{multirow}
\usepackage{wasysym}
\usepackage{multirow}
\usepackage{float}
\usepackage{pdfsync}
\usepackage{pstricks}
\newtheorem{theorem}{Theorem}
\newtheorem{proposition}{Proposition}
\newtheorem{proof}{Proof}
\newtheorem{definition}{Definition}
\newtheorem{property}{Property}

\newtheorem{corollary}{Corollary}
\doublespace

\begin{document}

\title{Throughput Analysis of Cognitive Wireless Networks with Poisson Distributed Nodes\\ Based on Location Information}
\author{Pedro H. J. Nardelli, Carlos H. Morais de Lima,\\ Hirley Alves, Paulo Cardieri, and Matti Latva-aho, 
\thanks{This research has been partly supported by: Infotech Graduate School at University of Oulu,  CNPq 312146/2012-4, the Brazilian Science without Boarders Special Visiting Researcher fellowship CAPES 076/2012, and SUSTAIN Finnish Academy and CNPq 490235/2012-3 jointly funded project.

Pedro H. J. Nardelli, Carlos H. Morais de Lima, Hirley Alves and Matti Latva-aho are with Center for Wireless Communications (CWC), University of Oulu, Finland (e-mail:~nardelli@ee.oulu.fi).
Carlos H. Morais de Lima is also with São Paulo State University, campus São João da Boa Vista, Brazil.
Paulo Cardieri is with the Wireless Technology Laboratory (WissTek), State University of Campinas, Brazil.}}

\maketitle
\markboth{Submission  v1: \today}{}

\begin{abstract}
This paper provides a statistical characterization of the individual achievable rates in bits/s/Hz and the spatial throughput of bipolar Poisson wireless networks in bits/s/Hz/m$^2$.
We assume that all transmitters have a cognitive ability to know the distance to their receiver's closest interferers so they can individually tune their coding rates to avoid outage events for each spatial realization.
Considering that the closest interferer approximates the aggregate interference of all transmitters treated as noise, we derive closed-form expressions for the probability density function of the achievable rates under two decoding rules: treating interference as noise, and jointly detecting the strongest interfering signals treating the others as noise.
Based on these rules and the bipolar model, we approximate the expected maximum spatial throughput, showing the best performance of the latter decoding rule.
These results are also compared to the reference scenario where the transmitters do not have cognitive ability, coding their messages at predetermined rates that are chosen to optimize the expected spatial throughput -- regardless of particular realizations -- which yields outages.
We prove that, when the same decoding rule and network density are considered, the cognitive spatial throughput always outperforms the other option.
\end{abstract}


\section {Introduction}
In the last few years, the demands for more efficient, reliable wireless systems induced network designers to think about alternative ways to supplement centralized cellular models.
One interesting idea is to build a multi-tier network where macro-base-stations coexist with a great number of smaller cells, which in turn operate in a more distributed fashion (e.g. the concept of femto-cell networks \cite{Chandrasekhar2008}).
Departing from the centralized approach whose capacity are fairly well characterized by Shannon theory, the limits of distributed systems that work in interference-limited regimes are unknown except for few specific cases, as discussed in \cite{andrews_rethinking_2008}.
In the following, we will discuss the main results on interference networks and how the concept of cognitive radio introduced in \cite{Haykin2005} is important in this context.

\subsection{Capacity of interference networks}
In 1978 Carleial formally stated the interference channel problem using arguments from information theory \cite{Carleial1978}.
Since then, several results have been proposed for the interference channel as discussed in  \cite[Ch. 6]{ElGamal_2012}.
Although these works shed light on the problem, even the capacity region of the simplest two-source-two-destination setting is still an open problem.
Moreover, when multiples sources and destinations are considered, such capacity regions becomes even more elusive.

Knowing such difficulties, some researchers have started investigating alternative approaches to better understand the limits of wireless networks with multiple communication pairs.
Gupta and Kumar introduced in \cite{Gupta2000} the transport capacity metric to determine how many bits-meter a wireless network with uniformly distributed nodes can reliably sustain when its density grows to infinite (asymptotic analysis).
After this milestone, many other papers have focused on a similar idea, finding the transport capacity scaling laws for different scenarios and under different assumptions.
The monograph \cite{Xue2006} compiles some of such studies.
Franceschetti et al. presented another important result in \cite{Franceschetti2009}, where they applied an unconventional method to find the physical limit of wireless networks by using laws of electrodynamics.
The authors further extended this approach in  \cite{Franceschetti2011} and determined the degrees of freedom of wireless networks based on electromagnetic theory.

Nevertheless both Franceschetti's and Gupta's lines of research strongly rely on asymptotic behaviors when the number of nodes infinitely grows, which  may  give an unclear picture of the actual  physical or medium access control network layers' design.
Bearing this aspect in mind, Weber et al. applied in \cite{WebYan2005} a statistical approach to characterize the throughput of wireless networks and then defined the transmission capacity as the highest spatial throughput\footnote{In the literature, spatial throughput can be also referred to as area spectral efficiency \cite{Alouini1999} or density of throughput \cite{baccelli2009_sto}} achievable without exceeding a maximum link outage probability, using the density of active links as the optimization variable.
An important aspect of this work is the use of stochastic geometry \cite{Baddeley_spatial_2007} to characterize the node spatial distribution as a Poisson point process (PPP).
Thereafter different strategies used in the wireless communications have been investigated such as interference cancellation, threshold transmissions, guard zones, bandwidth partitioning amongst others; the reference \cite{Weber2012} compiles these results.
In addition to them, we find in the literature other contributions using a similar approach.
For example, Vaze studied in \cite{Vaze2011} the throughput-delay-reliability trade-off in multi-hop networks using the metric random access transport capacity, which is an extension of the transmission capacity for multi-hop systems \cite[Sec. 4.2]{Weber2012}.
In \cite{Nardelli2013a}, the authors derived closed-form expressions for the throughput optimization under packet loss and queue stability constraints.
In \cite{Nardelli2011} a revisited version of the transmission capacity was proposed to compare different modulation-coding schemes.
The work \cite{Nardelli2012b} presents the transmission capacity optimization in term of the number of allowed retransmissions considering different medium access control protocols, which can be either synchronous or asynchronous.
Ganti et al. generalized in \cite{KrishnaGanti2011} the transmission capacity for different fading and node distributions for the high signal-to-interference regime.

Apart from these papers that focus on the statistic quantification of the spatial throughput of wireless networks, the use of models from stochastic geometry dates back the early 80's, when Takagi and Kleinrock firstly introduced the idea of evaluating the aggregate interference power of Poisson distributed interfering nodes \cite{Takagi1984}.
Thereafter, the subject has greatly developed and we can cite \cite{Baccelli2009_1,Baccelli2009_2,Haenggi,Cardieri2010} as relevant tutorials on how to apply stochastic geometry when analysing wireless systems.
As we will see later, this approach is important when dealing with cognitive networks, where self-organizing solutions are employed in a distributed manner.

\subsection{Complex systems and cognitive radio}
Let us start presenting a brief description of complex systems from the \cite{Page_2010}: ``A complex system consists of diverse entities that interact in a network or contact structure -- a geographic space, a computer network, or a market. 
These entities' actions are interdependent -- what one protein, ant, person, or nation does materially affects others.
In navigating within a complex system, entities follow rules, by which I mean prescriptions for certain behaviours in particular circumstances''.

For example, the tragedy of the commons problem described in \cite{Hardin1968} illustrates a counter-intuitive feature  of many independent and rational agents sharing a common pool of limited resources.
In this scenario, the agents optimize their own pay-offs in a selfish manner, i.e. find their individual global optimum, regardless of the others. 
Consequently, if every single agent takes the same rational decision, the shared resource will fade away after some time.
This problem is very context-dependent; for example, both fishing in a lake and forest harvesting can be  viewed as a tragedy of the commons class of problem, but the solution applied for each case tends to differ as the internal constraints of each system are different.
For wireless networks, the authors in \cite{Nardelli2013a} showed that the spatial throughput optimization under packet loss and queue stability constraints can be also viewed as a tragedy of the commons problem.

Another issue related to complex systems refers to the interplay between coordination and cooperation.
In game theory, the  prisoners' dilemma  is a good example of how coordination based on side information is important to optimize the system \cite{Leyton-Brown2008}.
In this game, rational agents, which cannot communicate to each other, should choose whether to cooperate or not.
If both cooperate, they get a higher pay-off than do not.
However, if one cooperate and the other does not, the non-cooperative agent will obtain a higher pay-off.
This fact leads to both agents not cooperating, which in turn provides lower pay-offs.
One interesting work was recently proposed by Nowak \cite{Nowak2012}, where the author describe different ways that cooperative behaviour can emerge in evolutionary systems.

Cooperative solutions are also important when dealing with co-channel interference in wireless networks.
For example, the authors in \cite{Pantisano13} employed game theory to build an algorithm to find coalitions of femto-cells that are willing to cooperate.
In \cite{Lima12} distributed coordination mechanisms were employed to control the aggregate interference level in stand-alone femto-cell networks.

Interestingly, these examples are based on self-organizing solutions, which refers to decentralized systems that are functional even without any central controlling entity (even though following interaction rules).
Many illustrations of this can be found in nature as, for instance, ants working in colonies, neurons building a capable brain etc. \cite{Page_2010}.
It is important to say that, different from these solutions that have emerged naturally, engineering systems do not accept outputs without a minimum quality requirement and therefore self-organization should be carefully designed, where the cognitive abilities and interaction rules should be well understood. 

Knowing the potential and the challenges of self-organization in engineering, Haykin proposed in his seminal work \cite{Haykin2005} the definition of cognitive radio: ``(...)  intelligent wireless communication system that is aware of its environment and uses the methodology of understanding-by-building to learn from the environment and adapt to statistical variations in the input stimuli, with two primary objectives in mind: highly reliable communication whenever and wherever needed; efficient utilization of the radio spectrum''.
This work indicates the direction to design more efficient wireless systems and thereafter the cognitive radio research have been rapidly growing.
%

\subsection{Contributions}
Motivated by the cognitive radio idea and the results presented in \cite{Nardelli2012} where the authors showed the importance of location information to improve the throughput of wireless networks, this paper focuses on studying wireless networks where every transmitter -- which are spatially distributed as a Poisson point process -- is able to use in a cognitive way the knowledge of its relative distances to the other transmitters for each different spatial realization. 
Following the results due to Baccelli et al. \cite{Baccelli2011}, we apply two different decoding rules: \emph{treating interference as noise} -- the IAN rule -- and \emph{joint detection of the strongest interferers' messages and treating the others as noise}\footnote{This rule splits the set of interferers into two mutually exclusive subsets: one contains the strongest interferers whose messages will be joint decoded with the desired one, and the other contains the transmitters with weaker detected power that will be treated as noise. This strategy is proved in \cite{Baccelli2011} the optimal for Gaussian point-to-point codes over interference channels, as discussed later on.} -- the OPT rule\footnote{We do not assume any interference cancellation (IC) technique as in \cite[Sec. 4.2]{Weber2012}, \cite{Mordachev2009,Vaze2012,Huang2008} since the OPT rule used in this paper always performs better than IC, as discussed in \cite{Baccelli2011,Blomer2009}.}.

Assuming that the aggregate interference can be approximated by the strongest interferer treated as noise, we derive an approximate probability density function (pdf) of the achievable rate in  bits/s/Hz that a typical link can sustain for the above decoding rules.
If the network follows the bipolar model \cite{Baccelli2009_1}\footnote{The details of this model will be described later on.}, the expected maximum spatial throughput of the network in bits/s/Hz/m$^2$ can be also approximated using those pdfs.

For comparison purposes, we consider the non-cognitive approach where transmitters use the same fixed coding rates (which is the most usual approach found in the literature, as in \cite{Weber2012}, \cite[Sec. IV]{Baccelli2011}, \cite{Blomer2009}), regardless of the specific spatial realization considered.
We then compute the highest spatial throughput for this setting by optimizing of the expected spatial throughput over different spatial realizations, where the optimization variable is the (symmetric) rate that the transmitters code their messages.
Differently from the former scenario where the coding rates are tuned to be the highest achievable ones given the relative nodes' positions for each different spatial realization, the fixed rate scheme only cares about the average behaviour of the network, resulting in decoding errors (outage events) for links whose capacity is below that predetermined rate.
We analytically prove that, under the same assumptions, the non-cognitive strategy always performs worse than the cognitive one.
Our numerical results confirm this difference and illustrate the advantages of using OPT over  IAN.

We also carry out an extensive simulation campaign to validate our findings and justify why our analysis is still relevant even when our approximations are loose -- although the closest-interferer approximation becomes looser when the network density grows, it follows that the qualitative relation and the quantitative ratio between the different strategies are maintained.
Besides we discuss the feasibility of the decoding rules and optimization strategies for different mobility patterns.
The cognitive approach is a feasible solution in (quasi-)static topologies, while the fixed rate optimization with IAN turns out to be the most appropriate choice in highly mobile topologies,.

The rest of this paper is divided as follows.
In Section \ref{sec_capacity_region}, we revisit the capacity region of Gaussian point-to-point codes over interference channels \cite{Baccelli2011} and then define the spatial throughput of wireless networks.
Section \ref{sec_spatial_Poisson} introduces the network model and the expected maximum spatial throughput using the cognitive approach.
Section \ref{sec_IAN} analyses the IAN decoder, while the OPT is the focus of Section \ref{sec_OPT}.
A comparison between the cognitive and the non-cognitive approaches is found in Section \ref{sec_Sym}.
We discuss both the accuracy of our approximations and implementation issues in Section \ref{sec_Discussions}, followed by the final remarks in Section \ref{sec_Final}.


\section{Capacity Region of Gaussian Point-to-Point Codes}
\label{sec_capacity_region}
This section reviews the capacity region of Gaussian point-to-point (G-ptp) codes for an arbitrary number of communication pairs as stated by Baccelli et al. in \cite[Sec. II]{Baccelli2011}.
For convenience let us assume a network with area $A$ [m$^2$] where $K+1$ source-destination pairs (also called transmitter-receiver pairs) coexist.
Each source node $i \in [0,\; K]$  transmits an independent message $M_i \in \left[1,\; 2^{n R_i}\right]$ to its respective destination $i$ at rate $R_i$ [bits/s/Hz], where $n$ is the codeword length.
Let $X_j$ be the complex signal transmitted by source $j \in [0,\; K]$ and let $Z_i \sim \mathcal{CN}(0,1)$, a the complex circularly symmetric Gaussian random variable with zero mean and unity variance, represent the noise effect at receiver $i$.
The detected signal $Y_i$  at receiver $i$ is then:
\begin{equation}
\label{eq_detected_signal}
    Y_i = \sum_{j=0}^K g_{ij} X_j + Z_i,
\end{equation}
where $g_{ij}$ are the complex channel gains between transmitter $j$ (TX$_j$) and receiver $i$ (RX$_i$).
We assume that every transmitted signal is subject to the same power constrain of $Q$ [W/Hz] such that the received signal between TX$_j$ and RX$_i$ is constrained by $P_{ij} = |g_{ij}|^2 Q$.

Each transmitter node uses a G-ptp code with a set of randomly and independently generated codewords $x_i^n(m_i) = (x_{i1},...,x_{in})(m_i)$ following  independent and identically distributed $\mathcal{CN}(0,\sigma^2)$ sequences such that $0< \sigma^2 \leq Q$, where $m_i \in \left[1,\; 2^{n R_i}\right] and \; i \in [0,\; K]$.
RX$_i$ receives a signal $y_i^n$ over the interference channel given by \eqref{eq_detected_signal} and then estimates the transmitted message as $\hat{m_i}(y_i^n) \in \left[1,\; 2^{n R_i}\right]$.
An error event in the decoding happens whenever the transmitted message differs from the estimated one.
Therefore the error probability of the G-ptp code is:
\begin{equation}
\label{eq_error_probability}
    p_n = \dfrac{1}{1+K} \; \sum_{i=0}^K \textup{Pr}[\hat{M_i} \neq M_i],
\end{equation}
where $Pr[\cdot]$ denotes probability that an event happens and $\hat{M}$ is the estimated message.

Next we use \eqref{eq_error_probability} to define the achievable rates and the capacity region for G-ptp codes.

\begin{definition}[achievable rates and capacity region]
\label{def_Achievable}
\emph{Let $\overline{p}_n$ be the average error probability  over G-ptp codes where $n$ is the codeword length.
Then, a rate tuple $\mathbf{R} = (R_0,...,R_K)$ is said to be achievable if $\overline{p}_n \rightarrow 0$ when $n \rightarrow \infty$.
In addition, the capacity region using G-ptp codes is the closure of the set of achievable  rate tuples $\mathbf{R}$.}
\end{definition}
 
This definition is important to define o establish the the capacity region of G-ptp codes as follows

\begin{theorem}[capacity region from \cite{Baccelli2011}]
\label{the_cap_reg}
Let $\mathcal{A}$ be the set of all $K+1$ transmitters in the network.
Let $\mathcal{A}_i$ denote a subset of $\mathcal{A}$ that contains TX$_i$ with $i  \in [0,\; K]$ and $\bar{\mathcal{A}}_i$ its complement.
The receiver of interest RX$_i$ then observes a multiple access channel (MAC) whose capacity region $\mathcal{R}_i$ is computed as
\begin{equation}
\label{eq_MAC_capacity_region}
    \mathcal{R}_i = \left\{ \mathbf{R}: \; \underset{k \in \mathcal{A}_i}{\sum} R_k \leq \log_2\left( 1 + \dfrac{\underset{k \in \mathcal{A}_i}{\sum} P_{ik}}{1 + \underset{j \in \bar{\mathcal{A}}_i}{\sum} P_{ij}} \right) \forall \;\; \mathcal{A}_i \subseteq \mathcal{A} \right\}.
\end{equation}

The capacity region $\mathcal{R}$ of the Gaussian interference channel with G-ptp codes is the intersection of the capacity regions $\mathcal{R}_i$ of all TX$_i$-RX$_i$ links with $i  \in [0,\; K]$, i.e.
\begin{equation}
\label{eq_capacity_region}
    \mathcal{R} =  \bigcap_{i=0}^{K} \mathcal{R}_i.  
\end{equation}
\end{theorem}

\begin{proof}
The proof of this theorem is found in \cite[Sec. II]{Baccelli2011}.
\end{proof}

This capacity region assumes a decoder that treats some of the interferers as noise, while others have their messages jointly decoded with the desired one.
In fact, this result is the basis of the OPT strategy mentioned in the previous section and further studied in Section \ref{sec_OPT}.
%

\section{Spatial Throughput of Bipolar Poisson Networks}
\label{sec_spatial_Poisson}
In this section, we will apply the results previously stated to establish an approximation for the spatial throughput of bipolar cognitive networks with transmitter nodes distributed according to a PPP.
But before that, let us define the spatial throughput and its maximum value using Theorem \ref{the_cap_reg} for a given spatial realization of the network as follows:

\begin{definition}[spatial throughput]
\emph{Let $A$ [m$^2$] be the network area and $K$ be the number of active links in $A$.
Then the spatial throughput, denoted by $\mathcal{S}$ and measured in bits/s/Hz/m$^2$, is defined as
\begin{equation}
\label{eq_spatial_rate}
    \mathcal{S} = \dfrac{1}{A} \; \sum_{i = 0}^{K} R_i.
\end{equation}
}
\end{definition}

\begin{definition}[maximum spatial throughput]
\emph{The maximum spatial throughput, denoted by $\mathcal{S}^*$, is defined as
\begin{equation}
\label{eq_spatial_capacity}
    \mathcal{S}^* = \max_{\mathbf{R}  \in  \mathcal{R}} \; \mathcal{S},
\end{equation}
such the rate tuple is achievable: $\mathbf{R} = (R_0,...,R_K)  \in  \mathcal{R}$.
}
\end{definition}

The maximum spatial throughput reflects the highest sum of achievable rates over a given area and it may vary depending on the network topology.
For example, clustered topologies (where transmitter--receiver pairs that are closer to each other, worsening the co-channel interference) tend to have  lower individual channel capacities than in more sparse ones, leading to different spatial throughputs even when the same area and number of links are considered.
To deal with this issue, we opt for studying Poisson distributed networks that are analytically tractable, allowing us to derive approximate expressions for $\mathcal{S}^*$ over different spatial realizations.

Let $\Phi$ be a two-dimensional homogeneous Poisson point process (PPP) with density $\lambda$ [nodes/m$^2$] that characterizes the spatial distribution of transmitters (TXs) over $\mathbb{R}^2$.
We assume that each TX is associated with one receiver (RX) located at a fixed distance $d$ from it in a random orientation\footnote{Note that the RXs are not part of the process  $\Phi$.} to establish a communication link; this is also known as Poisson bipolar model \cite{Baccelli2009_1}.
We consider that all TXs transmit information to their intended RXs over the same frequency band  (narrow-band) and using the G-ptp codes as described in Section \ref{sec_capacity_region}.

For each realization of  $\Phi$, the network may have a different capacity region $\mathcal{R}$ and consequently different maximum spatial throughputs $\mathcal{S}^*$.
When the network area is the infinite plane (i.e. $\mathbb{R}^2$), the capacity region given by equation \eqref{eq_capacity_region} becomes impossible to be computed\footnote{It is important to keep in mind that the number of links $K \rightarrow \infty$ when $A \rightarrow \mathbb{R}^2$.}.
Knowing these limitations, we choose to analyse the expected maximum spatial throughput, which allows us to evaluate the performance of bipolar Poisson networks over different spatial realizations of $\Phi$.
\begin{definition}[expected maximum spatial throughput]
\emph{Let $\mathbf{R}=(R_0,...,R_K)$ be a rate tuple  and $\mathcal{R}$ be the capacity region for a given network realization, then the expected maximum spatial throughput $\mathcal{C}$ is defined as
\begin{equation}
\label{eq_spatial_capacity_new}
    \mathcal{C} =  \mathbb{E}\left[ \mathcal{S}^* \right] = \mathbb{E}\left[ \max_{\mathbf{R}  \in  \mathcal{R}} \; \dfrac{1}{A} \; \sum_{i = 0}^{K} R_i\right],
\end{equation}
where $\mathbb{E}[\cdot]$ represents the expected value.}
\end{definition}

We can now apply properties from the point process theory \cite{Baddeley_spatial_2007} to approximate the average maximum spatial throughput for this class of Poisson networks as follows.

\begin{proposition}[expected maximum spatial throughput for bipolar Poisson networks]
\label{the_spatial_capacity}
For the bipolar Poisson network described in this section, the expected maximum spatial throughput $\mathcal{C}$ is given by:
\begin{equation}
\label{eq_spatial_capacity_poisson}
    \mathcal{C} \approx \lambda \; \mathbb{E}[R^*],
    \vspace{-1ex}
\end{equation}
where $\lambda$ is the network density and $R^*$ is the random variable that characterizes the  maximum spatial throughput achievable rates of a typical link over the network realizations.
\end{proposition}

\begin{IEEEproof}
Let us first remind that the spatial process $\Phi$ takes place in $\mathbb{R}^2$ and then $A \rightarrow \infty$, $K\rightarrow \infty$ and $\mathbf{R}=(R_0,R_1,...)$.
Then, we proceed with the following manipulation:
    \begin{align}
    \mathcal{C}
        & = \mathbb{E}\left[ \max_{\mathbf{R}  \in  \mathcal{R}} \; \lim_{A\rightarrow \infty} \dfrac{1}{A} \; \sum_{i = 0}^{\infty} R_i\right],& \\ \vspace{1ex}
        \label{eq_proof_approx_1}
        & \stackrel{(a)}{=} \mathbb{E}\left[\lim_{A\rightarrow \infty} \dfrac{1}{A} \; \sum_{i = 0}^{\infty} R^*_i\right],& \\ \vspace{1ex}
        \label{eq_proof_approx_2}
        & \stackrel{(b)}{\approx} \lambda \; \mathbb{E} [R^*].&
    \end{align}

Specifically, equality ($a$) considers the value of  $\mathbf{R}^* = (R_0^*,R_1^*,...) \in  \mathcal{R} $ that leads to the maximum spatial throughput for a given network realization, resulting in  $ \mathcal{S}^*$.
Since the PPP under analysis is homogeneous, we can apply Slivnyak theorem \cite[Ch. 3]{Baddeley_spatial_2007} to determine the statistical proprieties of any node in $\Phi$ over different spatial realizations based on a ``typical link'' -- a receiver node added at the origin, whose transmitter node is $d$ meters away from it.
Denoting the optimal coding rate employed by such a transmitter as $R^*$, we can make the approximation (b) by multiplying the network density $\lambda$ and $R^*$, which concludes this proof.
\end{IEEEproof}
\textbf{Remark:} Equality in $(b)$, instead of approximation in equation \eqref{eq_proof_approx_2},  is not possible since we cannot guarantee that the limit in equation \eqref{eq_proof_approx_1} exists.
It is also worth saying that, in this case, neither the spatial ergodic theorem nor the Campbell's theorem can be applied due to the interdependence between the elements of the optimal rate set $\mathbf{R}^*$ in each specific spatial realization.
In the following sections, we show that it is still possible to assess the performance of a typical link over different realizations based on closed-form expressions, which, we believe, makes valid our proposed approximation \eqref{eq_spatial_capacity_poisson}.

From equation \eqref{eq_spatial_capacity_poisson}, one can see that the main problem is now to derive the distribution of the maximum spatial throughput achieving rates $R^*$, which is our focus in the next two sections.
We would like to mention that Baccelli and Blaszczyszyn have presented in \cite[Sec. 16.2.3]{Baccelli2009_2} a general closed-form solution to the average rate of the typical link using Laplace transforms.
Nevertheless, we argue that our forthcoming derivations also contribute to the field due to their geometric appeal and simpler formulation, where we explicitly compute upper bounds on the Shannon rates of the typical link based on the distance from the typical receiver to its closest interferer that is treated as noise.

\section{Interference as Noise Decoding Rule}
\label{sec_IAN}

In this section we assess the decoding rule whereby the receivers treat the interference as noise -- or IAN decoders.
The following corollary shows its achievable rates.

\begin{corollary}[achievable rates for IAN decoders]
Assuming the noise is Gaussian and considering that TXs employ G-ptp codes as described in Section \ref{sec_capacity_region}, the rate $R_k$ associated with a given link  TX$_k$-RX$_k$ is achievable when IAN decoders are used if, and only if, the following inequality holds:
\begin{equation}
\label{eq_IAN_inequality}
    R_k \leq \log_2\left( 1 + \dfrac{P_{kk}}{1 + \underset{j \in \mathcal{A} \backslash \{ k \} }{\sum} P_{kj}} \right),
    \vspace{-1ex}
\end{equation}
where $ \mathcal{A}$ represents the set of active transmitters.
\end{corollary}

\begin{IEEEproof}
This is a special case of \eqref{eq_MAC_capacity_region} assuming that RX$_k$ only  decodes the message of TX$_k$ while the other TXs are treated as noise.
\end{IEEEproof}

We now apply this corollary to the scenario described in Section \ref{sec_spatial_Poisson} to assess the maximum expected spatial throughput of Poisson networks when receivers use IAN decoders.
Before we start, however, we still need to characterize the propagation phenomenon.
We consider here the distance-dependent path-loss model with exponent $\alpha>2$ \cite{yacoub_foundations_1993} so the channel gain between TX$_j$ and RX$_i$ is $|g_{ij}|^2 =  x_{ij}^{-\alpha}$, where $x_{ij}$ denotes the separation distance between them\footnote{This is in fact a simplified model that may lead to meaningless results for $x_{ij}<1$. As pointed in \cite{Inaltekin2009}, modified versions of this model just increase the complexity of the analysis without providing significant differences. We can also include into our channel modelling the effects of random fluctuations due to shadowing and multi-path as in \cite[Sec.4.1]{Weber2012}. For our purposes, though, the incorporation of these phenomena only complicates the mathematical formulation without giving any further insight on the network behaviour.}. 
We assume the noise power is negligible in comparison to the interference power (interference-limited regime).
We further consider that the aggregate interference experienced by RX$_k$ can be approximated by power $P_{k,\textup{clo}}$ related to its closest interferer.
Mathematically we have the following\footnote{This approximation is analysed in \cite{Mordachev2009} and it usually applied to compute lower bounds of the interference power based on  dominant interferers \cite{Weber2012,Nardelli2012}. 
We also discuss more about it in Section \ref{sec_Discussions}.}: $1 + \underset{j \in \mathcal{A} \backslash \{ k \} }{\sum} P_{kj} \approx P_{k,\textup{clo}}$.
Based on these assumptions, we can derive an approximation of the pdf of the highest achievable rate of the typical link when IAN decoders are employed.

\begin{proposition}[approximate pdf of the highest achievable rates for IAN]
\label{prop_pdf_IAN}
The pdf of highest rate $R^*$ achieved by the typical link can be approximated by 
\begin{equation}
\label{eq_pdf_R_IAN}
    f_{R^*}(x) \approx \ln 4  \; \dfrac{ 2^x \lambda \pi d^2 \left(2^x -1\right)^\frac{2}{\alpha }}{\alpha\left(2^x -1\right)} \; e^{ - \lambda \pi d^2 \left(2^x -1\right)^\frac{2}{\alpha}},
\end{equation}
where $x>0$.
\end{proposition}

\begin{IEEEproof}
Let us analyze a typical link TX$_0$--RX$_0$ added to the PPP $\Phi$.
From Slivnyak theorem (refer to \cite[Th. 3.1]{Baddeley_spatial_2007}), this inclusion does not affect the distribution of $\Phi$.
Without loss of generality, we assume that the origin of the plane is located at RX$_0$ and label the interferers TX$_i$ accordingly to their distances to RX$_0$, i.e. TX$_1$ is the closest, TX$_2$ is the second closest and so on.
From our assumptions, we have $1 + \sum_{k=1}^\infty P_{k} \approx P_1$.
We then apply the path-loss model to the IAN decoder presented in equation \eqref{eq_IAN_inequality}, considering that the distances from TX$_0$ and TX$_1$ to RX$_0$ are respectively $d>0$ and $r_1>0$, resulting in
\begin{equation}
\label{eq_distance_IAN}
    R_0 \leq \log_2\left( 1 + \dfrac{d^{-\alpha}}{r_1^{-\alpha}} \right),
\end{equation}
where $r_1$ is a random variable.

To compute the pdf of $r_1$, we use the definition of contact zone \cite[Defs. 1.9 and 3.2]{Baddeley_spatial_2007} (the distance between a typical point and its first neighbor), resulting in \cite{haenggi_distances_2005} 
\begin{equation}
\label{eq_pdf_r1_IAN}
    f_{r_1}(x) = 2 \lambda \pi x  e^{ -\lambda \pi x^2},
\end{equation}
such that $x>0$.
Defining $\beta_0^*=d^{-\alpha}/r_1^{-\alpha}$ such that inequality \eqref{eq_distance_IAN} still holds, then we have the following relation  $r_1 = d \beta_0^{*\frac{1}{\alpha}}$ (see Fig. \ref{FigureIllustration}).
We now apply this variable transformation to \eqref{eq_pdf_r1_IAN} and hence the pdf of $\beta^*_0 >0$ can be obtained as 
\begin{equation}
\label{eq_pdf_beta_IAN}
    f_{\beta^*_0}(x) = \dfrac{2 \lambda \pi d^2 x^{\frac{2}{\alpha}}}{\alpha x}  e^{- \lambda \pi d^2 x^{\frac{2}{\alpha}}},
\end{equation}
where $x>0$.

\begin{figure}
    \centering
    \includegraphics[width=0.35\columnwidth]{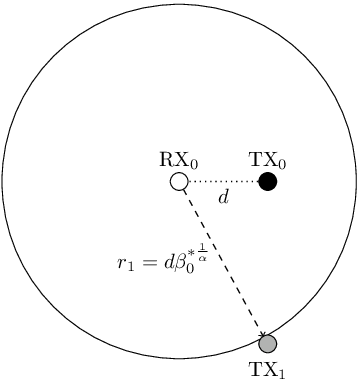}
    \caption{Illustrative example of the typical link TX$_0$--RX$_0$ employing the IAN decoding rule, where TX$_1$ represents the closest interferer to RX$_0$. To reach the highest achievable rate $R_0^*$, the relation $r_1= d \beta_0^{*\frac{1}{\alpha}}$ must be respected such that $r_1$ is the random variable that represents the distance between RX$_0$ and TX$_1$.}
\label{FigureIllustration}
\end{figure}

To conclude this proof, we proceed with the transformation $R^*_0 = \log_2(1+\beta^*_0)$ remembering that PPPs are stationary so we can characterize any node of the network based on a typical node, dropping the index $0$ (refer to \cite[Sec. 3.4]{Baddeley_spatial_2007}).
\end{IEEEproof}

\textbf{Remark:} This maximum value can be achieved only when TX$_0$ knows the distance $r_1$ for each different spatial realization.
Our result consider that TX$_0$ implements a cognitive solution to first acquire local network topology and then use it as side information so as to set its coding rate to be achievable based on the propagation model and the defined TX$_0$-RX$_0$ distance $d$.

The result just stated provides us an approximation\footnote{We discuss the tightness of the closest-interferer approximation later in Section \ref{sec_Discussions}.}  of pdf for IAN decoders over an infinite plane and over different spatial realizations of the process $\Phi$.
Then, we apply \eqref{eq_pdf_R_IAN} to approximate the expected maximum spatial throughput given by \eqref{eq_spatial_capacity_poisson}, resulting in
\begin{equation}
\label{eq_spatial_capacity_IAN}
    \mathcal{C}_\textup{IAN} \approx \lambda \; \int_0^\infty x f_{R^*}(x) \; dx,
\end{equation}
which does not have a closed-form solution and a numerical integration is required.
For this reason, next we derive some proprieties\footnote{Such properties rely on the closest interferer approximation that will be discussed later on. For simplicity we hereafter refer to the approximate expected maximum spatial throughput as cognitive spatial throughput.} of \eqref{eq_spatial_capacity_IAN} that help us to understand the $\mathcal{C}_\textup{IAN}$ behavior.

\begin{property}[concavity of the cognitive spatial throughput]
\label{prop_concavity_IAN}
A function $f(\cdot)$ is said to be quasi-concave if, and only if, $  f\left(p x_1 + (1-p) x_2 \right) \geq \min\{ f(x_1), f(x_2)\}$, where $0\leq p \leq 1$.
Considering that the rate that leads to the cognitive spatial throughput, $R^*$, is a function of the network density $\lambda$ (i.e. $R^*=f(\lambda)$), then
$\mathcal{C}_\textup{IAN}$ given by \eqref{eq_spatial_capacity_IAN} is quasi-concave in terms of $\lambda$, where $R^*$ is a random variable characterized by the pdf \eqref{eq_pdf_R_IAN}.
\end{property}

\begin{IEEEproof}
Let us first consider two different network densities $\lambda_1$ and $\lambda_2$ such that $\lambda_1 < \lambda_2$.
Then, defining that $\lambda = p \lambda_1 + (1-p) \lambda_2$ with $0\leq p \leq 1$,  we proceed with the following manipulation
\begin{align}
    \mathcal{C}_\textup{IAN}(\lambda) & = \left(p \lambda_1 + (1-p) \lambda_2 \right) \;  \mathbb{E}[f\left(p \lambda_1 + (1-p) \lambda_2 \right)] &\\ \vspace{1ex}
    & \stackrel{(a)}{\geq} \;\; \lambda_1 \; \mathbb{E}[f\left(p \lambda_1 + (1-p) \lambda_2 \right)] & \\ \vspace{1ex}
    & \stackrel{(b)}{=} \;\; \lambda_1  \; \mathbb{E}[f(\lambda_1)] \; = \; \mathcal{C}_\textup{IAN}(\lambda_1) &  \\ \vspace{1ex}
    & \stackrel{(c)}{\geq} \;\; \lambda_2 \; \mathbb{E}[f(\lambda_2)] \; = \; \mathcal{C}_\textup{IAN}(\lambda_2). & \vspace{1ex}
\end{align}

Inequality $(a)$ comes from the fact that $\lambda_1 \leq p \lambda_1 + (1-p) \lambda_2$ whereas equality $(b)$ is obtained by setting $p=1$ since the first inequality holds for all $0\leq p \leq 1$.
This proves the quasi concavity of the analyzed function when $\lambda_1 \textup{E}[f(\lambda_1)] < \lambda_2 \textup{E}[f(\lambda_2)]$.
Finally, inequality $(c)$ is straightforward  when $\lambda_1 \textup{E}[f(\lambda_1)] \geq \lambda_2 \textup{E}[f(\lambda_2)]$, which concludes this proof.
\end{IEEEproof}

\begin{property}[highest cognitive spatial throughput]
\label{prop_opt_spatial_capacity_IAN}
The network density $\lambda^*$ that leads to the cognitive spatial throughput given by \eqref{eq_spatial_capacity_IAN} is obtained as the density $\lambda>0$ which is solution to the following equation:
\begin{equation}
\label{eq_lambda_opt_IAN}
    \int_0^\infty x^{\frac{2}{\alpha}-1} \log_2(1+x)    e^{ - \lambda \pi d^2 x^{\frac{2}{\alpha}}} \; d x = \int_0^\infty x^{\frac{2}{\alpha}-1} \left( \lambda \pi d^2 x^{\frac{2}{\alpha}} - 1 \right)  \log_2(1+x)    e^{ - \lambda \pi d^2 x^{\frac{2}{\alpha}}} \; d x.
\end{equation}
\end{property}

\begin{IEEEproof}
Let us first rewrite the cognitive spatial throughput formulation using the pdf $f_{\beta^*}(x)$ given by \eqref{eq_pdf_beta_IAN}, yielding
\begin{equation}
\label{eq_spatial_capacity_IAN_beta}
    \mathcal{C}_\textup{IAN} = \lambda \; \int_0^\infty \log_2 (1+x) f_{\beta^*}(x) \; dx.
\end{equation}

As shown in Property \ref{prop_concavity_IAN}, the $\mathcal{C}_\textup{IAN}$ is quasi-concave in terms of $\lambda$ so we find its maximum value based on the derivative equation $d\mathcal{C}_\textup{IAN}/d\lambda = 0$.
After some algebraic manipulation, we obtain \eqref{eq_lambda_opt_IAN}, which concludes this proof.
\end{IEEEproof}

\begin{property}[lower bound]
\label{prop_lower_IAN}
A lower bound of the cognitive spatial throughput given by \eqref{eq_spatial_capacity_IAN} is computed as
\begin{equation}
\label{eq_capacity_lower_IAN}
    \mathcal{C}_\textup{IAN} \geq \lambda  y  e^{ - \lambda \pi d^2 (2^y -1)^\frac{2}{\alpha}},
\end{equation}
where $y > 0$.
\end{property}

\begin{IEEEproof}
To prove this property, we apply the Markov inequality as presented below:
\begin{equation}
\label{eq_capacity_lower_IAN_proof}
    \textup{Pr}[ R^* \geq y ] \leq \dfrac{\mathbb{E}[R^*]}{y} \Rightarrow \mathbb{E}[R^*] \geq y e^{ - \lambda \pi d^2 (2^y -1)^\frac{2}{\alpha}},
\end{equation}
where $\textup{Pr}[ R^* \geq y ] = 1 - \int_0^y f_{R^*}(x) \;dx$ and $2^y - 1 > 0$.

Then, we multiply both sides by $\lambda$, resulting in \eqref{eq_capacity_lower_IAN}.
\end{IEEEproof}

\begin{property}[upper bound]
\label{prop_upper_IAN}
An upper bound of the cognitive spatial throughput given by \eqref{eq_spatial_capacity_IAN} is computed as
\begin{equation}
\label{eq_capacity_upper_IAN}
    \mathcal{C}_\textup{IAN} \leq \lambda \log_2\left( 1 + \left(\dfrac{1}{\lambda \pi d^2}\right)^\frac{\alpha}{2} \Gamma\left(1 + \dfrac{\alpha}{2}\right)\right).
\end{equation}
where  $\Gamma(\cdot)$ is the Euler gamma function  defined as $\Gamma(z) = \int_0^\infty t^{z-1} e^{-t} \; dt$.
\end{property}

\begin{IEEEproof}
Let us apply Jensen's inequality based on the concavity of \eqref{eq_spatial_capacity_IAN} (refer to Property \ref{prop_concavity_IAN}), yielding
\begin{align}
\label{eq_proof_upper}
    \mathcal{C}_\textup{IAN}
        & =  \lambda \; \mathbb{E}[  R^* ] & \\
        & \stackrel{(a)}{=} \lambda  \; \mathbb{E}[  \log_2(1+\beta^*) ] & \\
        \label{eq_capacity_upper_IAN2}
        & \stackrel{(b)}{\leq} \lambda \; \log_2(1+\mathbb{E}[\beta^*]),  &
\end{align}
where  equality $(a)$ comes from the fact that $R^*=\log_2(1+\beta^*)$ and  inequality $(b)$ is the Jensen inequality for quasi-concave functions.
Then, we compute the expectation of the random variable $\beta^*$ using \eqref{eq_pdf_beta_IAN}, which proves \eqref{eq_capacity_upper_IAN}.
\end{IEEEproof}

\begin{property}[asymptotic equivalence]
\label{prop_asymp_IAN}
Let $\sim$ denote the asymptotic equivalence of two functions, then
\begin{equation}
\label{eq_asymp}
    \mathcal{C}_\textup{IAN} \sim c \; \lambda^{1 - \frac{\alpha}{2}},
\end{equation}
when $\lambda \rightarrow \infty$ and $c = \left(\dfrac{1}{\pi d^2}\right)^\frac{\alpha}{2} \Gamma\left(1 + \dfrac{\alpha}{2}\right)$.
\end{property}

\begin{IEEEproof}
To prove that two functions $f(x)$ and $g(x)$ are asymptotically equivalent, i.e.  $f(x) \sim g(x)$, we should show that $\displaystyle \lim_{ x\rightarrow \infty} f(x)/g(x) = 1$.
Let us first consider the behavior of the random variable $\beta^*$, characterized by \eqref{eq_pdf_beta_IAN} when $\lambda \rightarrow \infty$, yielding
\begin{equation}
\label{eq_limit1}
    \lim_{\lambda \rightarrow \infty} f_{\beta^*}(x) = \delta(x),
\end{equation}
where $\delta(x)$ is the Dirac impulse function. 

This indicates that the random variable $\beta^*$ tends to have the value $0$ with high probability when the network density increases. Now, let us consider that $\beta^* \rightarrow 0$, then we have the following limit
\begin{equation}
\label{eq_limit2}
    \lim_{\beta^* \rightarrow 0} \dfrac{\log_2(1+\beta^*)}{\beta^*} = \dfrac{1}{\ln 2}.
\end{equation}

Using these limits and recalling \eqref{eq_pdf_beta_IAN}, we can manipulate the expression of the cognitive spatial throughput $\mathcal{C}_\textup{IAN}$ as follows.
\begin{equation}
\label{eq_limit3}
    \lim_{\lambda \rightarrow \infty} \mathcal{C}_\textup{IAN} = \lim_{\lambda\rightarrow \infty} \lambda \; \mathbb{E} [\log_2(1+\beta^*)] = \lim_{\lambda\rightarrow \infty} \lambda \; \dfrac{\mathbb{E}[\beta^*]}{\ln 2}.
\end{equation}

Proceeding similarly with the upper bound in \eqref{eq_capacity_upper_IAN2}, we have
\begin{equation}
\label{eq_limit4}
    \lim_{\lambda\rightarrow \infty} \lambda \; \log_2(1+\mathbb{E} [\beta^*]) = \lim_{\lambda\rightarrow \infty} \lambda  \; \dfrac{\mathbb{E}[\beta^*]}{\ln 2}.
\end{equation}

Now, we recall that the division of limits is the limit of the division, resulting in
\begin{equation}
\label{eq_limit5}
    \lim_{\lambda\rightarrow \infty} \dfrac{\lambda \; \mathbb{E}[\log_2(1+\beta^*)]}{\lambda \; \log_2(1+\mathbb{E} [\beta^*])} = 1.
\end{equation}

From this fact, we can state from \eqref{eq_capacity_upper_IAN} that
\begin{equation}
\label{eq_asymp_new}
    \mathcal{C}_\textup{IAN} \sim \lambda  \log_2\left( 1 + \left(\dfrac{1}{\lambda \pi d^2}\right)^\frac{\alpha}{2} \Gamma\left(1 + \dfrac{\alpha}{2}\right)\right),
\end{equation}
when $\lambda \rightarrow \infty$.

To conclude this proof, we verify that $\left(\dfrac{1}{\lambda \pi d^2}\right)^\frac{\alpha}{2} \Gamma\left(1 + \dfrac{\alpha}{2}\right) \rightarrow 0$ when $\lambda \rightarrow \infty$ and then we apply the approximation $\log(1+x) \approx x$ valid when $x\rightarrow 0$ into \eqref{eq_asymp_new} resulting \eqref{eq_asymp}.
\end{IEEEproof}

\begin{figure}[t]
    \centering
    \includegraphics[width=0.7\linewidth]{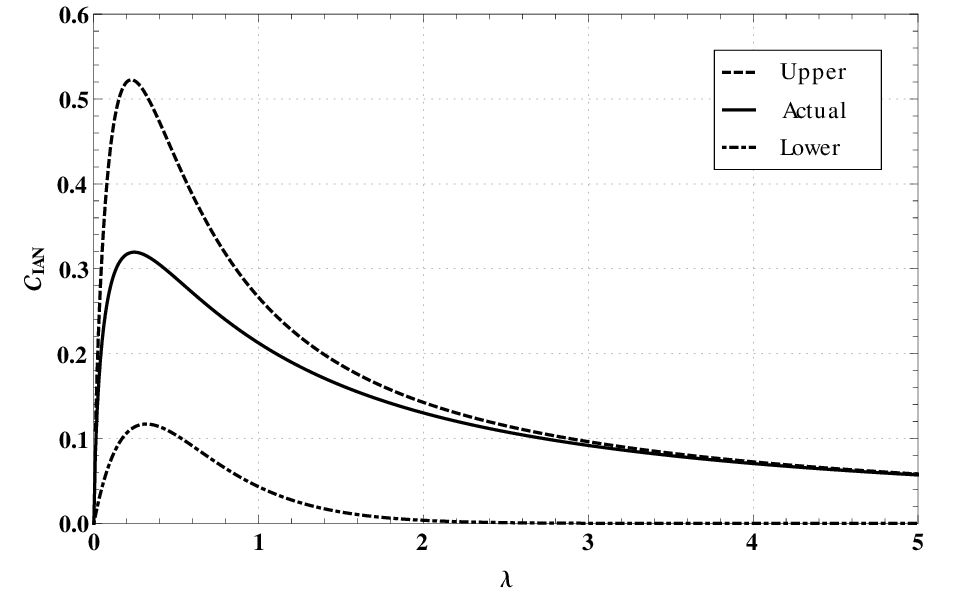}
    \caption{Actual values, lower and upper bounds of the cognitive spatial throughput, $\mathcal{C}_\textup{IAN}$, versus the network density $\lambda$ for $\alpha=4$ and $d=1$. The lower bound is obtained using  $y=1$ in equation \eqref{eq_capacity_lower_IAN}. The actual values and upper bound are computed using equations \eqref{eq_spatial_capacity_IAN} and \eqref{eq_capacity_upper_IAN}, respectively.}
\label{FigureAsymptotic_IAN}
\end{figure}

Fig. \ref{FigureAsymptotic_IAN} illustrates the behavior of the cognitive spatial throughput $\mathcal{C}_\textup{IAN}$ and its proposed bounds as a function of the network density $\lambda$.
Firstly, one can notice that the cognitive spatial throughput has a maximum point which is expected from its concavity stated in Property \ref{prop_concavity_IAN} and the density $\lambda^*$ that achieves the optimal is given by equation \eqref{eq_lambda_opt_IAN}\footnote{A closed-form solution is unknown for this equation  but standard numerical methods solve it. In our case, we use {\sf FindRoot} from Wolfram Mathematica 9.}.
When densities lower than this maximum are considered, the network is spatially not saturated in terms of interference and the cognitive spatial throughput of the network is still not in its highest value.
In this situation, any increase of $\lambda$ leads to an increase of $\mathcal{C}_\textup{IAN}$ until the  inflexion point is achieved.
After that point,  the network spatial throughput degrades due to the proximity of the interferers, strongly reducing the average of the link rates $R^*$.
Consequently, $\mathcal{C}_\textup{IAN}$  becomes a decreasing function of $\lambda$.

From Fig. \ref{FigureAsymptotic_IAN}, we can also evaluate the proposed upper and lower bounds of the cognitive spatial throughput.
As one can notice the lower bound proposed in Property \ref{prop_lower_IAN} is loose, regardless of $\lambda$.
In fact, the main use for this bound is to prove the relation between the cognitive spatial throughput and the maximum spatial throughput achieved with fixed rates, as it will be discussed later on.
Regarding Property \ref{prop_upper_IAN}, when $\lambda$ increases, the upper bound become tighter, as predicted by Property \ref{prop_asymp_IAN}.
In other words, the upper bound has the same value as the cognitive spatial throughput $\mathcal{C}_\textup{IAN}$ when $\lambda \rightarrow \infty$ as shown in Fig. \ref{FigureAsymptotic_IAN}.
In the next section, we apply the same approach used here to derive the cognitive spatial throughput and its properties when OPT decoders are considered.

\section{Optimal Decoding Rule}
\label{sec_OPT}

As previously discussed, the optimal decoding strategy when Gaussian point-to-point  codes are used in wireless networks with multiple transmitter-receiver pairs consists in jointly decoding some messages from the strongest interferers, while the rest is treated as noise.
Based on this observation, we obtain the achievable rates for links whose receivers use OPT decoders as follows.
\begin{corollary}[achievable rates for OPT decoding rule]
Assuming that Gaussian noise and considering that TXs use the Gaussian point-to-point codes as described in Section \ref{sec_capacity_region}, then the rate $R_k$ associated with a given link  TX$_k$-RX$_k$ is said to be achievable when the OPT decoder is employed if, and only if, the following inequality holds:
\begin{equation}
\label{eq_OPT_inequality}
      R_k \leq \log_2\left( 1 + \dfrac{\underset{i \in \mathcal{A}^*_k}{\sum} P_{ki}}{1 + \underset{j \in \bar{\mathcal{A}^*_k}}{\sum} P_{kj}} \right) - \underset{i \in \mathcal{A}^*_k \backslash \{ k \}}{\sum} R_i,
\end{equation}
where $\mathcal{A}^*_k$ represents the subset of transmitters whose messages are decoded by receiver $k$ and $\mathcal{A}^*_k \cup \bar{\mathcal{A}^*_k} = \mathcal{A}$ is the set of all active transmitters in the network.
\end{corollary}

\begin{IEEEproof}
To obtain \eqref{eq_OPT_inequality} we proceed with a simple manipulation of equation \eqref{eq_MAC_capacity_region}, isolating the rate $R_k$ related to  TX$_k$-RX$_k$ link by considering the subsets $\mathcal{A}^*_k$ that lead to achievable rates.
\end{IEEEproof}

Next we will apply the theorem stated above to statistically characterize the achievable rates over different spatial realizations using the OPT decoding rule and then approximate the expected maximum spatial throughput of the network described in Section \ref{sec_spatial_Poisson}, which is given by equation \eqref{eq_spatial_capacity_new}.
Under the assumption of OPT decoding rule, however, the analysis is more complicated since the receiver node should choose  the subset of messages that will be jointly decoded and then verify whether the coding rate of its own transmitter is achievable, given all other coding rates.
By construction, all receivers proceed in the same way and hence the analysis becomes a very intricate combinatorial problem.
For this reason, we need to approximate the pdf of the highest achievable rates for the OPT decoders, we resort to some assumptions that will be justified afterwards.

As before, we only consider the deterministic path-loss (refer to Section \ref{sec_IAN}) and that the sum of the interfering signals observed by RX$_k$ that are treated as noise can be approximated by the signal from the closest interferer amongst them, whose power is denoted $P_{k,\textup{clo}}$.
If the noise power is negligible compared to $P_{k,\textup{clo}}$, then $1 + \underset{j \in \bar{\mathcal{A}^*}_k }{\sum} P_{kj} \approx P_{k,\textup{clo}}$.
Based on these assumptions, we can state the following proposition.

\begin{proposition}[approximate pdf of the highest achievable rates for OPT]
\label{prop_pdf_OPT}
Let us denote the rate tuple that achieves the maximum expected spatial throughput for the network described in Section \ref{sec_spatial_Poisson} as $\mathbf{R}^* = (R_0^*,R_1^*,...) \in  \mathcal{R}$.
If the pdf of $R^*_k, \forall \; k \in \mathcal{A}$ follows the pdf of a typical rate $R^*$ and denoted by $f_{R^*}(x)$, then
\begin{equation}
\label{eq_pdf_R_OPT_general}
    f_{R^*}(x) \approx \sum_{i=0}^\infty \dfrac{(\lambda \pi d^2 )^i}{\Gamma(i)} \; e^{-\lambda \pi d^2} f_{R^*|n}(x|n=i)
	\vspace*{-1ex}
\end{equation}
where $f_{R^*|n}(x|n)$ is the pdf of $R^*$ given that $1+n$ messages are jointly decoded and is approximated by
\begin{equation}
\label{eq_pdf_R_OPT}
    f_{R^*|n}(x|n) \approx \ln 4  \; \dfrac{2^{(1+n)x} \lambda \pi d^2 }{\alpha} \; \left(\frac{2^{(1+n)x} -1}{1+n}\right)^{\frac{2}{\alpha}-1} \; e^{ - \lambda \pi d^2 \left(\left(\frac{2^{(1+n)x} -1}{1+n}\right)^\frac{2}{\alpha} -1\right)},
    \vspace{-1ex}
\end{equation}
such that $x>\frac{\log(2+n)}{1+n}$. 
\end{proposition}

\begin{IEEEproof}
Let us first deal with the typical link TX$_0$--RX$_0$.
Without loss of generality, we place the origin of the Cartesian plane at RX$_0$ and assume that all nodes that are closer to RX$_0$ than TX$_0$ have their messages jointly decoded with TX$_0$ message  (see Fig. \ref{FigureIllustration2}). 
From the distance-dependent path-loss model, the closer the TX, the higher the power, and then this choice of the subset $\mathcal{A}^*_0$ is justified.

\begin{figure}
\centering
\includegraphics[width=0.35\columnwidth]{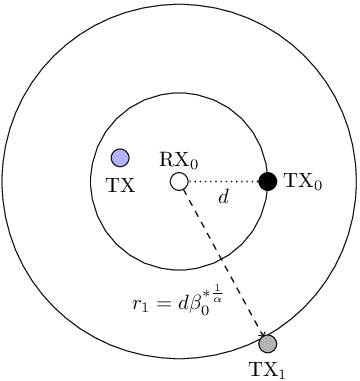}
\caption{Illustrative example of the typical link TX$_0$--RX$_0$ employing the OPT decoding rule. The blue TX has its message jointly decoded with TX$_0$ message and TX$_{1}$ is the closest interferer to RX$_0$ whose signal is treated as noise. The random variable $r_1$ denotes the distance between RX$_0$ and TX$_1$ such that $r_1>d$.}
\label{FigureIllustration2}
\end{figure}

For each network spatial realization, we consider that a number $n$ associated with the transmitters whose messages are decoded by RX$_0$ is known, which yields the following inequality
%
\begin{equation}
\label{eq_OPT_inequality2}
    \log\left( 1 + \dfrac{(1+n) P_{00}}{P_{0,\textup{clo}}} \right) < \log\left( 1 + \dfrac{P_{00} + \sum_{i=1}^n P_{0i}}{P_{0,\textup{clo}}} \right).
\end{equation}
One can observe from \eqref{eq_OPT_inequality} and \eqref{eq_OPT_inequality2} that rate tuples that satisfy $R_{0} + \sum_{i=1}^n R_{i} < \log\left( 1 + \frac{(1+n) P_{00}}{P_{0,\textup{clo}}} \right)$ are always achievable.
Defining $\beta^*_0 = P_{00}/P_{0,\textup{clo}}$, we use similar steps to the ones used in the proof of Proposition \ref{prop_opt_spatial_capacity_IAN}, but considering now that $r_1>d$ to compute the pdf $ f_{\beta^*_0}(x)$ as 
\begin{equation}
\label{eq_pdf_beta_OPT}
    f_{\beta^*_0}(x) = \dfrac{2 \lambda \pi d^2 x^{\frac{2}{\alpha}}}{\alpha x}  e^{- \lambda \pi d^2 \left(x^{\frac{2}{\alpha}}-1\right)},
    \vspace{-2ex}
\end{equation}
where $x > 1$ and $ f_{\beta^*_0}(x) = 0$ when $x \leq 1$.

Then, we assume that $R_0 + \sum_{i=1}^n R_{i} \approx (1+n) R_0$ to obtain $(1+n) R_0^* = \log\left( 1+ (1+n) \beta_0^* \right)$.
By applying such a transformation, we can find the pdf of $R^*_0$ given $n$.
From the assumption that the all links perform similar to the typical one, we can drop the index $0$, resulting in equation \eqref{eq_pdf_R_OPT}.
To unconditioned the pdf $f_{R^*|n}(x|n)$, we compute the probability that there exist$n=i$ points of the PPP in the area $\pi d^2$, concluding this proof.
\end{IEEEproof}

\textbf{Remark:} In addition to the closest interferer treated as noise approximation, this proposition is based on other two strong assumptions: (i) the detected power at RX$_0$ related to the $1+n$ jointly decoded messages is equal to $(1+n) P_{00}$  and (ii)  the sum rate associated with those messages is given by $(1+n) R_0$.
Assumption (i) uses the lower bound given by \eqref{eq_OPT_inequality2}, which indicates that we  underestimate the aggregate power and (ii)  approximates the sum of $1+n$ random variables that follows the same distribution by one random variable multiplied by $1+n$.
We argue that the underestimation by-product of (i) leaves us some room for variations in the sum rate approximation used in (ii).
In addition, due to the homogeneity of the spatial process, $R_0 + \sum_{i=1}^n R_{i} \approx (1+n) R_0$  leads to a reasonable approximation.
Simulations results are presented in Section \ref{sec_Discussions} where we discuss such approximations.

Here we  approximate the expected maximum spatial throughput\footnote{As in the previous section we use the term cognitive spatial throughput to refer to the approximate expected maximum spatial throughput.} $\mathcal{C}_\textup{OPT}$  when the OPT decoding rule is employed as
\begin{equation}
\label{eq_spatial_capacity_OPT}
    \mathcal{C}_\textup{OPT} \approx \lambda \; \sum_{i=0}^\infty \dfrac{(\lambda \pi d^2 )^i}{\Gamma(i)} \; e^{-\lambda \pi d^2} \int_{\frac{\log(2+i)}{1+i}}^\infty x \; f_{R^*|n}(x|n=i) \; dx,
\end{equation}
where $f_{R^*|n}(x|n=i)$ is given in Proposition \ref{prop_pdf_OPT}.

The integral in \eqref{eq_spatial_capacity_OPT} is analytically unsolvable (we can rely on numerical solutions, though).
To gain more insights on the system performance, we next derive some properties of the cognitive spatial throughput.

\begin{property}[concavity]
\label{prop_concavity_OPT}
Considering that the rate  $R^*$ is a function of the network density $\lambda$, then
$\mathcal{C}_\textup{OPT}$ given by \eqref{eq_spatial_capacity_OPT} is \emph{quasi-concave} in terms of $\lambda$, where $R^*$ is a random variable given by  \eqref{eq_pdf_R_OPT_general}.
\end{property}

\begin{property}[lower bound]
\label{prop_lower_OPT}
A lower bound of the cognitive spatial throughput given by \eqref{eq_spatial_capacity_OPT} is computed as
\begin{equation}
\label{eq_capacity_lower_OPT}
    \mathcal{C}_\textup{OPT} \geq \lambda \; \sum_{i=0}^\infty \dfrac{(\lambda \pi d^2 )^i}{\Gamma(i)}  y  e^{- \lambda \pi d^2\left(\frac{2^{(1+i)y}-1}{1+i} \right)^\frac{2}{\alpha}},    %
\end{equation}
where $y > \frac{\log_2(2 +i)}{1+i}$ for all $i \geq 0 $.
\end{property}

\begin{property}[upper bound]
\label{prop_upper_OPT}
A upper bound of the cognitive spatial throughput given by \eqref{eq_spatial_capacity_OPT} is computed as
\begin{equation}
\label{eq_capacity_upper_OPT}
    \mathcal{C}_\textup{OPT} \leq \lambda \; \sum_{i=0}^\infty \dfrac{(\lambda \pi d^2 )^i}{\Gamma(i)}  \; \dfrac{e^{-\lambda \pi d^2}}{1+i} \; \log_2\left( 1 + (1+i) \left(\dfrac{1}{\lambda \pi d^2} \right)^{\frac{2}{\alpha}} \Gamma\left( 1 + \dfrac{2}{\alpha}, \lambda \pi d^2 \right) e^{\lambda \pi d^2}    \right),
\end{equation}
where  $\Gamma(\cdot, \cdot)$ is the incomplete Gamma function, which is defined as $\Gamma(z,a) = \int_a^\infty t^{z-1} e^{-t} \; dt$.
\end{property}
\begin{property}[asymptotic equivalence]
\label{prop_asymp_OPT}
Let $\sim$ denote asymptotic equivalence of two functions, then
\begin{equation}
\label{eq_asymp_OPT}
    \mathcal{C}_\textup{OPT} \sim \lambda \; \sum_{i=0}^\infty \dfrac{(\lambda \pi d^2 )^i}{\Gamma(i)}  \; \dfrac{e^{-\lambda \pi d^2}}{1+i} \; \log_2\left( 1 + (1+i) \left(\dfrac{1}{\lambda \pi d^2} \right)^{\frac{2}{\alpha}} \Gamma\left( 1 + \dfrac{2}{\alpha}, \lambda \pi d^2 \right) e^{\lambda \pi d^2}    \right),
\end{equation}
when $\lambda \rightarrow \infty$.
\end{property}

The proofs of these properties follow the same principles used in the previous section so we do not present them here.
It is worth pointing though out that the proofs of \eqref{eq_capacity_lower_OPT}-\eqref{eq_asymp_OPT} begin by assuming that the number $1+n$ of jointly decoded messages is known.
Then, we use the fact that the unconditioned the cognitive spatial throughput is a linear combination of  the conditioned cognitive spatial throughputs with weights given by the Poisson probabilities that $n=i$ nodes lie in a area of $\pi d^2$, given by the probability $\frac{(\lambda \pi d^2 )^i}{\Gamma(i)}  \; e^{-\lambda \pi d^2}$.

Fig. \ref{FigureAsymptotic_OPT} presents the  cognitive spatial throughput $\mathcal{C}_\textup{OPT}$ given by \eqref{eq_spatial_capacity_OPT} as a function of $\lambda$ together with its proposed upper and lower bounds.
One can observe that the lower bound given by Property \ref{prop_lower_OPT} is very loose for the value of the constant $y$ that was arbitrarily chosen ($y=2$).
This bound, however, can be improved by tuning the constant $y$ in accordance to the number of jointly decoded messages.
Such an improvement in the proposed bound will be discussed in the next section when we apply it to analytically assess the performance of networks where predetermined fixed rates are imposed.

\begin{figure}[t]
\centering
\includegraphics[width=0.7\linewidth]{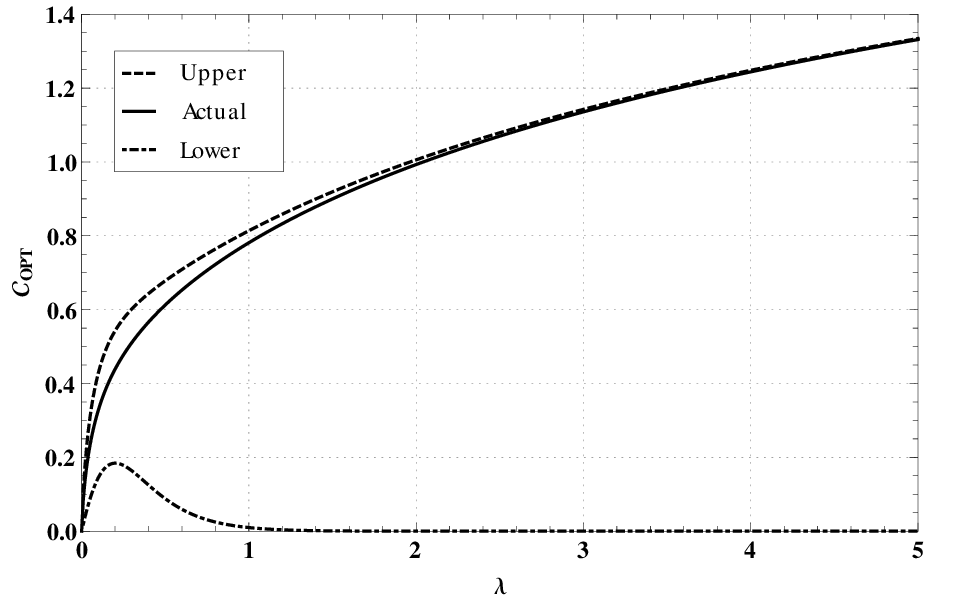}
\caption{Actual values, lower and upper bounds of the cognitive spatial throughput, $\mathcal{C}_\textup{OPT}$, versus the network density $\lambda$ for $\alpha=4$ and $d=1$. The lower bound is obtained using $y=2$ in equation \eqref{eq_capacity_lower_OPT}. The actual values and upper bound are computed using equations \eqref{eq_spatial_capacity_OPT} and \eqref{eq_capacity_upper_OPT}, respectively.}
\label{FigureAsymptotic_OPT}
\end{figure}
Turning your attention  to the values of $\mathcal{C}_\textup{OPT}$  given by \eqref{eq_spatial_capacity_OPT}, one can easily see that it is an increasing function of $\lambda$.
For lower densities, $\mathcal{C}_\textup{OPT}$ increases faster since the probability that an interfering TX has its message jointly decoded is also low and, consequently, the rate is constrained by the interferers that are treated as noise, indicating that $\mathcal{C}_\textup{OPT}$ is limited by the low spatial reuse.
When $\lambda$ increases, on the other hand, more messages from interfering TXs start being jointly decoded, which diminishes the $\mathcal{C}_\textup{OPT}$ rate of increase.
Furthermore, we can observe that the upper bound proposed in Property \ref{prop_upper_OPT} is a good approximation to $\mathcal{C}_\textup{OPT}$ for all densities $\lambda$ especially when $\lambda \rightarrow \infty$, corroborating Property \ref{prop_asymp_OPT}.

By comparing the results shown in Fig. \ref{FigureAsymptotic_IAN} (IAN) and Fig. \ref{FigureAsymptotic_OPT} (OPT), one can see that the OPT decoding rule provides higher cognitive spatial throughputs, regardless of the network density.
The performance gain obtained with the OPT decoder  indicates that the mechanism of joint detection used here is a good way to cope with the strongest interferers.
A more detailed comparative analysis between OPT and IAN decoding rules is presented later.

In the following section, we compare the results obtained so far with the non-cognitive approach:  coding rates are fixed for a given network density and set to optimize the average spatial throughput, regardless of a specific network topology.
In this way, the transmitters do not use the local knowledge of the network topology as side information, leading to outage events  (i.e. some pairs use coding rates above their channel capacity).

\section{Spatial Throughput Optimization Using Predetermined Fixed Rates}
\label{sec_Sym}
We now focus our attention on  scenarios where TXs, which do not have access to location information, set their coding rates to the fixed values that leads to the highest expected spatial throughput, given, however, that the TXs are aware of how many messages are jointly decoded by their RXs.
Using this scheme, groups of TXs use the same fixed coding rates and then an optimization problem is formulated to find these rates such that the expected spatial throughput is maximized.
As a result the optimal choice of coding rates, as discussed later on, is outside the network capacity region, stated in Theorem \ref{the_cap_reg}, leading to outage events for some links.
Next, we define of the aforementioned optimization problem.

\begin{definition}[highest expected spatial throughput]
\emph{The expected spatial throughput optimization problem for a network where TXs have fixed coding rates is defined as
\vspace{-2ex}
\begin{equation}
\label{eq_transm_capacity1}
    \mathcal{T} = \max_{\mathbf{R}} \; \mathbb{E}[\mathcal{S}],
    \vspace{-1ex}
\end{equation}
where $\mathcal{T}$ is the maximum expected spatial throughput, $\mathbf{R}=(R_0, R_1,...)$ represents the set of fixed coding rates $R_i$ used by the TXs such that $i$ is the number of jointly decoded messages in addition to the desired one, and $\mathcal{S}$ is the spatial throughput given by \eqref{eq_spatial_rate}, where only  the successful transmissions are taking into account.}
\end{definition}

When the IAN decoding rule is used, there is no jointly decoded message and then the optimization is only related to one fixed coding rate\footnote{This is the usual approach as in \cite{Weber2012},\cite{Blomer2009}.}.
We now present two propositions that state the highest expected spatial throughputs for IAN and OPT decoders applying the network model used before\footnote{Once again we use the closest intereferer treated as noise approximation.}.

\begin{proposition}[highest expected spatial throughput for IAN]
\label{pro_sym_IAN}
The highest expected spatial throughput $\mathcal{T}_\textup{IAN}$ achieved when IAN decoders are used is given by
\vspace{-2ex}
\begin{equation}
\label{eq_sym_IAN}
    \mathcal{T}_\textup{IAN} =  \lambda \log_2(1+\beta^*) e^{-\lambda \pi d^2 \beta^{*\frac{2}{\alpha}}},
    \vspace{-1ex}
\end{equation}
where $\beta^*$ is  the value of $\beta>0$, which is solution of
\vspace{-1ex}
\begin{equation}
    \label{eq_beta_sym_IAN}
    \beta= \left( \dfrac{2}{\alpha} \lambda \pi d^2 (1+\beta) \ln(1+\beta)  \right)^\frac{\alpha}{\alpha-2}.
    \vspace{-1ex}
\end{equation}
\end{proposition}
\begin{IEEEproof}
Let us first rewrite the expected  spatial throughput given by \eqref{eq_spatial_rate} for this scenario as
\vspace{-2ex}
\begin{equation}
 \label{eq_spacap_sym}
    \mathcal{S} = \lambda R P_\textup{s},
    \vspace{-1ex}
\end{equation}
where $R$ is the fixed coding rate used by all TXs and $P_\textup{s}$ is the corresponding success probability.

We proceed here similarly to the proof of Proposition \ref{prop_pdf_IAN} and then apply the relation  $R=\log_2(1+\beta)$, where $R, \beta>0$.
From the closest  interferer assumption, an outage event occurs whenever an interfering TX node lies inside the area defined by the circumference centered at the RX node and with radius $d\beta^\frac{1}{\alpha}$ (see Fig. \ref{FigureIllustration}).
Using the Poisson distribution, we have that $P_\textup{s} = e^{-\lambda \pi d^2 \beta^{\frac{2}{\alpha}}}$.
Hence, we can rewrite equation \eqref{eq_spacap_sym}  as
\vspace{-2ex}
\begin{equation}
 \label{eq_spacap_sym_IAN}
    \mathcal{S} = \lambda \log_2(1+\beta) e^{-\lambda \pi d^2 \beta^{\frac{2}{\alpha}}},
    \vspace{-2ex}
\end{equation}
which is a concave function of $\beta$.

Hence, we compute $\beta^*$ which is the solution of the derivative equation $d\mathcal{S}/d\beta = 0$, resulting after some manipulation in \eqref{eq_beta_sym_IAN}.
To conclude this proof, we use $\beta^*$ into equation \eqref{eq_spacap_sym_IAN}, obtaining \eqref{eq_sym_IAN}.
\end{IEEEproof}

\begin{proposition}[highest expected spatial throughput for OPT]
\label{pro_sym_OPT}
The highest expected spatial throughput $\mathcal{T}_\textup{OPT}$ achieved when OPT decoders are used is given by
\vspace{-1ex}
\begin{equation}
\label{eq_sym_OPT}
    \mathcal{T}_\textup{OPT} =  \lambda \; \sum_{i=0}^\infty \dfrac{(\lambda \pi d^2 )^i}{\Gamma(i)}  \; \dfrac{e^{-\lambda \pi d^2}}{1+i} \; \log_2(1 + (1+i)\beta_i^*) \;  e^{- \lambda \pi d^2 \left(\beta_i^{*\frac{2}{\alpha}}-1\right)}
	\vspace{-1ex}
\end{equation}
where, $\beta^*_i$ is found as the value of $\beta_i>1$ for $i \in \mathbb{N}$, which is solution of
\vspace{-1ex}
\begin{equation}
\label{eq_beta_sym_OPT}
    \beta_i = \left(\dfrac{2}{(1+i) \alpha} \; \lambda \pi d^2 (1 + (1+i) \beta_i) \; \ln(1+(1+i)\beta_i) \right)^{\frac{\alpha}{\alpha-2}}.
    \vspace{-1ex}
\end{equation}
\end{proposition}

\begin{IEEEproof}[Outline of proof]
To prove this proposition, we follow the same steps used in the proof of Proposition \ref{pro_sym_IAN}, considering these basic differences: $\beta_i = d^{-\alpha}/r_1^{-\alpha}>1$ (since messages from TXs closer to a given RX than its own TX are jointly decoded and then $r_1>d$) and the optimization is proceeded for each $i=0,1,2,...$ which yields \eqref{eq_beta_sym_OPT}.
To conclude this outline, we average the expected spatial throughputs by the Poisson probabilities that  $i$ nodes lie  in the area $\pi d^2$, resulting in \eqref{eq_sym_OPT}.
\end{IEEEproof}

Here we apply Properties \ref{prop_lower_IAN} and \ref{prop_lower_OPT} to obtain an analytical relation between the expected highest spatial throughput $\mathcal{C}$ (cognitive) and the highest expected spatial throughput $\mathcal{T}$ (non-cognitive) using fixed rates for either decoding rules.

\begin{proposition}[$\mathcal{C}$  vs. $\mathcal{T}$]
\label{prop_spacap_vs_sym}
For a given density $\lambda$ and assuming that all links use the same decoding rule (either IAN or OPT), then
\vspace{-1ex}
\begin{equation}
\label{eq_spacap_vs_sym}
	\mathcal{C} \geq \mathcal{T}.
    \vspace{-2ex}
\end{equation}
\end{proposition}
\begin{IEEEproof}
This statement is a consequence of Property \ref{prop_lower_IAN}, when we set the constant $y = \log(1+\beta^*)$ in \eqref{eq_capacity_lower_IAN}, yielding \eqref{eq_sym_IAN}.
Similarly, we use Property \ref{prop_lower_OPT}, applying for each different $i\in \mathbb{N}$ a different constant $y$ in \eqref{eq_capacity_lower_OPT} such that $y_i =  \frac{\log(1+(1+i)\beta^*_i)}{1+i}$, which yields \eqref{eq_sym_OPT},  concluding this proof.
\end{IEEEproof}

Fig. \ref{FigureSym_Rates} shows the maximum expected spatial throughput following the formulation derived in this section.
As proved in Proposition \ref{prop_spacap_vs_sym}, $\mathcal{T}$ is always lower than or equal to $\mathcal{C}$ for the same density and the same decoding rule.
This is justified by the methodology used to derive the cognitive spatial throughput, which allows for a choice of coding rate based on the location information for each different realization.
When fixed rates are used, the transmitters code their messages using a fixed rate that depends only on the number of other messages that are jointly decoded by their own receivers.
By optimizing based only on the average behavior, some RXs cannot successfully decode their messages for specific topologies, which decreases the expected spatial throughput.
Therefore, the cognitive strategy has always the best performance.
Besides given the decoding rule employed, the curves of $\mathcal{T}$ and $\mathcal{C}$ have a similar shape.

\begin{figure}[t]
\centering
\includegraphics[width=0.7\linewidth]{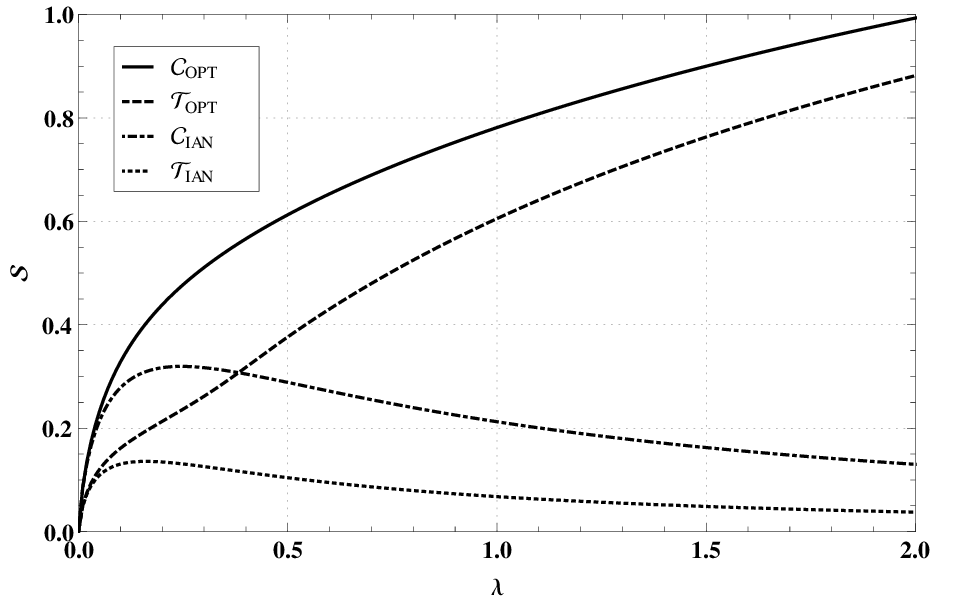}
\caption{The highest expected spatial throughputs $\mathcal{T}$ using fixed coding rates given by \eqref{eq_sym_IAN} and \eqref{eq_sym_OPT}, and the cognitive spatial throughputs $\mathcal{C}$ given by \eqref{eq_spatial_capacity_IAN} and \eqref{eq_spatial_capacity_OPT} as a function of the network density $\lambda$ for IAN and OPT decoding rules, $d=1$ and $\alpha=4$.}
\label{FigureSym_Rates}
\end{figure}
Fig. \ref{FigureSym_Rates} also shows that the cognitive spatial throughput obtained when OPT is used has  a huge gain if compared with the IAN option.
This result reflects that the OPT rule is able to avoid the strongest interferers by jointly decoding their messages.
When the density $\lambda$ is low, both OPT and IAN decoders have approximately the same performance since the probability that a interferer is closer to a given RX than its own TX is very low.
Increasing $\lambda$, such a probability also increases and then the differences between the strategies become apparent as the closest interferer is the limiting factor for IAN, while such node may have its message jointly decoded when OPT is used, what decreases the harmful effects of the nearby interferers.
%

\vspace{-2ex}
\section{Discussions}
\label{sec_Discussions}
So far we have showed that, for same network density, OPT decoders outperform IAN, and  the cognitive strategy outperforms the non-cognitive one when receivers employ the same decoding rule.
Nevertheless we still need to discuss some possible limitations of our finds, namely the tightness of our approximations and the feasibility of each decoding rule for practical implementations.
In the following subsections we deal with both aspects, identifying why our results are important even when our approximation is poor and for which circumstances the design setting that provides the worst performance is more suitable than the optimal.

\vspace{-2ex}
\subsection{Tightness of our approximation}
\label{subsec_tightness}
Here we discuss the validity of the ``closest interferer treated as noise approximation'' used to derive the approximate performance of both decoding rules.
Figure \ref{Figure_simulation} shows the cognitive spatial throughput  $\mathcal{C}$ computed using our analytical approximation and Monte Carlo simulation as a function of the network density $\lambda$ for both decoders\footnote{The results for the  highest expected spatial throughput presented in Section \ref{sec_Sym} follows the same trends and thus they are not presented here.}.
For both IAN and OPT, the lower the density is, the better our approximation works.
Conversely, increasing the density, our approximate spatial throughput gets looser and looser.
\begin{figure}[t]
\centering
\includegraphics[width=0.7\linewidth]{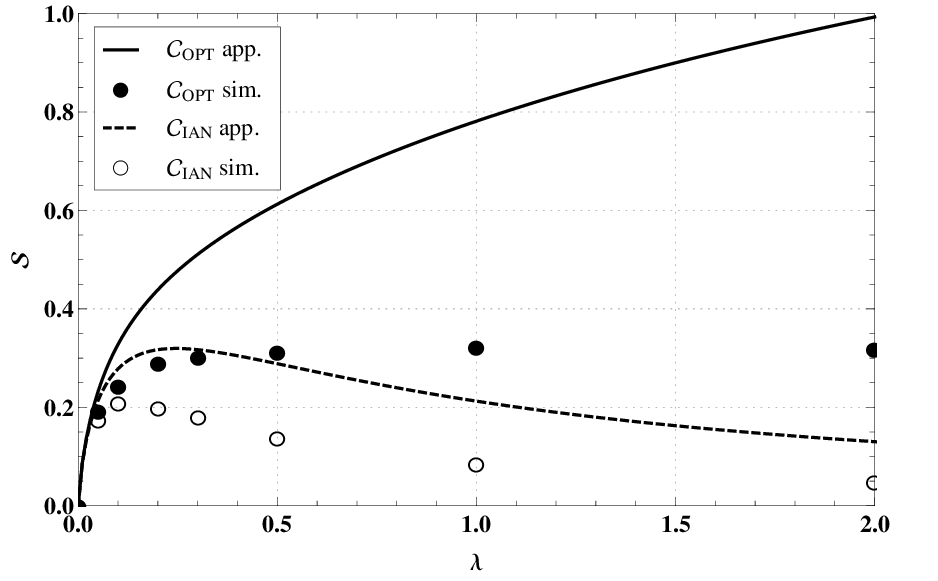}
\caption{Cognitive spatial throughputs  $\mathcal{C}$ for IAN and OPT as a function of the network density $\lambda$ , $d=1$ and $\alpha=4$. Approximate results given by equations \eqref{eq_spatial_capacity_IAN} and \eqref{eq_spatial_capacity_OPT}, and simulations.}
\label{Figure_simulation}
\end{figure}
The closest-interferer approximation is in fact a lower bound of the aggregate interference \cite{Weber2012}, leading then to an upper bound of the actual cognitive spatial throughput.
This bound have been proved to be asymptotically equivalent to the actual values when  $\lambda \rightarrow 0$\cite{Mordachev2009,Weber2012}\footnote{In our point of view this asymptotic analysis is unsuitable for the study carried out here; we assume an interference-limited network, which opposes the idea of very low density of interferers. When $\lambda \rightarrow 0$, we see the network in its noise-limited regime.}.
For higher densities, the closest interferer treated as noise tends to contribute less to the aggregate interference experienced by the receivers, worsening our approximation.
Besides, we obtained our numerical results using the path-loss exponent $\alpha=4$ and Weber et al. showed that lower exponents lead to looser bounds  \cite{WebYan2005}.

Yet, we believe that the comparison between the IAN and OPT decoders is fair since the results presented in Sections \ref{sec_IAN} and \ref{sec_OPT} rely on the same approximation\footnote{We can argue in the same way to say that the analysis presented in Section \ref{sec_Sym} is also fair.}.
We further argue that our approximation has no effect on the trade-off analysis done in this paper and Figure \ref{Figure_simulation} illustrates this fact by showing that the OPT always outperforms IAN in similar scales: the ratios  $\mathcal{C}_\textup{IAN}/\mathcal{C}_\textup{OPT}$ obtained via simulation or via our approximations have similar values when considering the same $\lambda$.
As the proposed formulation provides a computationally simpler way to assess the network performance than numerical simulations, we reinforce the contribution of this paper even when our approximations provide less accurate bounds. 

All in all, we believe that our main messages -- OPT is better than IAN, and cognitive strategy is better than the non-cognitive one -- are unaffected by our approximations, which are shown by both qualitative relations and quantitative ratios between our analytical and simulated results.
Despite of these facts, the optimal strategy is infeasible for practical implementation as discussed in the following.

\vspace{-2ex}
\subsection{Design setting and mobility pattern}
\label{subsec_mobility}
Throughout this paper we have shown that the best design option in terms of spatial throughput is to employ OPT decoders and apply the cognitive scheme.
This solution, however, has downsides: (i) RXs require the knowledge of the codebooks of the jointly decoded messages and (ii) OPT decoders are computationally more complex than IAN.

Knowing that, we argue that the use of either/both OPT and/or cognitive strategy is infeasible for (highly) mobile topologies.
Under this topology, the neighbours of any given receiver change very fast, rendering the joint decoding procedure impossible.
Shopping malls and streets where people  move frequently can exemplify this scenario.
If this is the case, even though the configuration employing IAN decoders with fixed rate optimization is far from the optimal performance, it is a more suitable choice.

Conversely, when (quasi-)static networks are considered, the optimal strategy becomes viable.
In this case, receiver nodes must known the codebooks of their strongest interfering nodes and jointly decode their messages. 
In addition, the links must coordinate their coding rates to be in the network capacity region. 
Smart homes, industry plants and other kind of machine-to-machine communications can exemplify this mobility pattern.

Besides, there are other aspects that may be prohibitive for OPT. 
For instance, many applications require secrecy and then the codebook knowledge makes OPT infeasible even for static topologies.
Other applications need fast processing time, which is also infeasible when many interfering messages are jointly decoded.
Anyway, this dependence of the topology must be taken into account when the network is designed.
Furthermore, the mobility pattern of the network can also change over time -- for example, offices during the night are quasi-static, while highly dynamic during parts of the working hours.

All these aspects indicates the needs for ad hoc adaptive algorithms that estimate the network state and proceed with their optimization according to their cognitive ability.
If the closest interferer treated as noise approximation gives a reliable indication, the results presented herein might even provide a practical way of implementing them.

\section{Final Remarks}
\label{sec_Final}
In this paper we studied the spatial throughput of cognitive networks using the Gaussian point-to-point codes, where transmitter nodes use the location information of their receiver's closest interferer to tune their coding rates.
Assuming that the network follows the bipolar Poisson model, we evaluated two different decoding rules: (i) treat all interfering messages as noise -- IAN, and (ii) jointly decode the messages whose detected power is higher than the desired message power while treat the remaining as noise -- OPT.

We proposed an approximation of the expected highest spatial throughput for Poisson distributed networks where transmitter nodes are able to cognitively tune their coding rates for each spatial realization based on the location information of the closest interferer of their respective receiver.
We then stated several properties of our approximation using either decoders and showed that, when  the same network density and decoding rules are assumed, the cognitive strategy always outperforms the non-cognitive one, where transmitters code their messages in order to optimize the expected spatial throughput using pre-determined fixed rates, regardless of a specific network realization.

These results can be actually used to implement an ad hoc algorithm capable to adapt the coding rates based on estimated information about distances, network density and mobility.
In fact, we have already identified the work done in \cite{PROC:LIMA-ITW09} as a potential starting point to further develop the theory presented here to more practical scenarios.

\bibliographystyle{IEEEtran}
\bibliography{ref_abbrev}

\end{document}